\documentclass[conference]{IEEEtran}

\usepackage{amsmath,amssymb,amsthm}
\usepackage{graphicx}
\usepackage{booktabs}
\usepackage{algorithm}
\usepackage{algorithmic}
\usepackage{cite}
\usepackage{xcolor}
\usepackage[caption=false,font=footnotesize]{subfig}
\usepackage{url}
\usepackage[T1]{fontenc}
\usepackage[utf8]{inputenc}
\usepackage[hidelinks]{hyperref}

\newtheorem{theorem}{Theorem}
\newtheorem{lemma}[theorem]{Lemma}
\newtheorem{proposition}[theorem]{Proposition}
\newtheorem{corollary}[theorem]{Corollary}

\newtheorem{remark}[theorem]{Remark}

\newcommand{\nqubits}{n}
\newcommand{\shots}{s}

\title{Polynomial-Resource Classification of Quantum Circuit Families via Classical Shadows}

\author{
\IEEEauthorblockN{Andrew Maciejunes}
\IEEEauthorblockA{
Virginia Modeling, Analysis, and Simulation Center\\
Old Dominion University\\
Norfolk, VA, USA\\
amacieju@odu.edu
}
\and
\IEEEauthorblockN{Ross Gore}
\IEEEauthorblockA{
Center for Secure and Intelligent Critical Systems\\
Old Dominion University\\
Norfolk, VA, USA\\
rgore@odu.edu
}
\and
\IEEEauthorblockN{Sachin Shetty}
\IEEEauthorblockA{
    Center for Secure and Intelligent Critical Systems \\
    Professor, Department of Computer and Electrical Engineering \\ 
    Old Dominion University \\
    Norfolk, VA, USA \\
    sshetty@odu.edu
}
\and
\IEEEauthorblockN{Barry Ezell}
\IEEEauthorblockA{
    Virginia Modeling, Analysis, and Simulation Center \\
    Old Dominion University \\
    Norfolk, VA, USA \\
    bezell@odu.edu
}
}

\begin{document}

\maketitle

\begin{abstract}
We compare four polynomial-resource measurement strategies, (I) $Z$-basis-only, (II) nearest-neighbor $ZZ$ (NN), (III) multi-basis ($Z$, $X$, $Y$), and (IV) classical shadows, for classifying three quantum circuit families: IQP, Clifford, and Clifford$+T$.
We find $Z$-only measurements outperform multi-basis and classical shadows across all qubit counts and all four classifiers evaluated, and the $O(\nqubits)$-feature NN strategy matches $Z$-only to within $0.02$ in Random Forest accuracy.
The best result is a Random Forest accuracy of $0.91$ at 4--5 qubits under $Z$-only ($0.89$ for NN, $0.85$ for multi-basis, $0.67$ for shadows).
All four strategies collapse to near-chance accuracy ($\approx 0.33$) above approximately 12 qubits under the quadratic shot budget $\shots = 16\nqubits^2$.
These findings indicate that the discriminative signal between these circuit families is concentrated in local, nearest-neighbor $Z$-basis correlations, consistent with the diagonal gate structure of IQP circuits, and that additional Pauli correlator types or long-range correlations carry no compensating discriminative power for this task.
We support these findings with two exact results: a frame identity implying that every purely-$X$ observable of an IQP output state vanishes identically, and an estimator-variance analysis showing that, at equal shot budget, the multi-basis and classical-shadow protocols estimate each $ZZ$ correlator with exactly $3\times$ and at least $9\times$ the variance of $Z$-only measurement, respectively.
These results establish that a quadratic shot budget is insufficient for reliable classification above approximately 12 qubits, but do not rule out the existence of a subquadratic or otherwise more efficient polynomial-resource strategy; whether any polynomial measurement protocol can classify these families at large qubit counts remains an open question.
These results motivate a formal open question: whether any BPP algorithm can reliably distinguish IQP-generated distributions from classically simulable families at large qubit counts.
\end{abstract}


\section{Introduction}
\label{sec:intro}

A fundamental question in quantum information is whether the output distributions of quantum circuits, distributions that cannot be efficiently represented classically, can be meaningfully characterized using only polynomially many measurements \cite{montanaro2013survey}. 
Quantum state tomography, originally proposed by Vogel and Risken (1989) \cite{vogel1989determination} and first demonstrated experimentally by Smithey et al.\ (1993) \cite{smithey1993measurement}, provides a complete reconstruction of the quantum state but at prohibitive cost. 
Haah et al.\ (2016) \cite{haah2016sample} rigorously established this cost with a lower bound of $\Omega(d^2/\delta)$ copies, where $d$ is the Hilbert space dimension and $\delta$ is the desired precision. 
Since $d = 2^n$ for an $n$-qubit system, full tomography requires exponentially many measurements in the number of qubits, rendering it infeasible for systems of practical interest. 
This exponential barrier motivates the search for protocols that extract meaningful information about quantum states without full reconstruction.

Quantum circuit families provide a principled testbed for this question.
IQP, Clifford, and Clifford$+T$ circuits produce distributions with qualitatively different statistical signatures, rooted in their distinct computational complexity properties.
Clifford circuits are classically simulable; IQP sampling is believed to be classically hard under plausible complexity-theoretic conjectures~\cite{bremner2016average}; Clifford$+T$ circuits, a universal gate set \cite{boykin1999universal, bravyi2005universal}, interpolate between the two.
If these families can be distinguished with polynomially many measurements, it suggests that their distributional signatures are encoded in low-order statistics that a polynomial measurement protocol can access.

The key methodological challenge is choosing the right measurement protocol.
Measuring in a single basis ($Z$-only) gives access to $ZZ$ correlators but misses the $X$ and $Y$ components of the Pauli correlation structure.
Measuring in three global bases ($Z$, $X$, $Y$) recovers $ZZ$, $XX$, and $YY$ correlators but requires $3\times$ the shot budget and still misses mixed-basis correlators ($XY$, $XZ$, $YZ$).
Classical shadows resolve this. By applying a random single-qubit Clifford rotation before measurement, all six 
two-qubit Pauli correlators can be estimated from a single unified protocol with sample complexity 
$O(\log n^2) = O(\log n)$ to predict all pairwise correlators simultaneously, significantly more efficient 
than the $O(n^2)$ shots required to estimate each correlator type independently.
In this work, all four measurement strategies share the same $O(\nqubits^2)$ shot budget to ensure a fair experimental comparison; under this budget classical shadows accesses all six Pauli correlator types simultaneously, while $Z$-only, NN, and multi-basis access fewer.

\paragraph*{Contributions}
\begin{itemize}
    \item An empirical comparison of four measurement strategies ($Z$-only, nearest-neighbor $ZZ$, multi-basis, and classical shadows) for classifying IQP, Clifford, and Clifford$+T$ circuit families, showing that $Z$-only and NN outperform multi-basis and shadows and inverting the hypothesis that richer Pauli correlation structure would benefit classification.
    \item Evidence that the discriminative signal is concentrated in local, nearest-neighbor $ZZ$ correlations: an $O(\nqubits)$-feature NN strategy matches the $O(\nqubits^2)$-feature $Z$-only strategy to within $0.02$ in Random Forest accuracy, consistent with the diagonal computational-basis structure of IQP circuits.
    \item Characterization of a scaling collapse near 12 qubits at which all implemented polynomial strategies degrade to near-chance performance under the quadratic shot budget, establishing a practical limit for this feature set and shot scaling law.
    \item Two exact theoretical results (Lemma~\ref{lem:frame}, Propositions~\ref{prop:xnull} and~\ref{prop:variance_ratio}, and Corollary~\ref{cor:variance_ordering}): every purely-$X$ observable of an IQP output state is identically zero, and at equal shot budget the multi-basis and shadow protocols estimate each $ZZ$ correlator with exactly $3\times$ and at least $9\times$ the variance of $Z$-only measurement. Together these explain the observed strategy ordering without unproven assumptions.
\end{itemize}

\section{Background}
\label{sec:background}

\subsection{Circuit Families as Distribution Classes}

We consider three quantum circuit families that span a range of computational complexity.

\paragraph{Clifford Circuits}
Generated by $\{H, S, \mathrm{CNOT}\}$, Clifford circuits are efficiently simulable classically via the stabilizer formalism~\cite{gottesman1997stabilizer}; this gate set is provably minimal and sufficient for generating the full Clifford group in any finite dimension~\cite{farinholt2014ideal}.
Their output distributions possess structured pairwise correlations reflecting the underlying stabilizer group geometry.

\paragraph{Clifford$+T$ Circuits}
Adding $T = \mathrm{diag}(1, e^{i\pi/4})$ yields a universal gate set~\cite{boykin1999universal}.
Sufficient $T$-gate density breaks the stabilizer structure and moves the circuit outside the efficiently simulable regime, producing a distributional signature intermediate between Clifford and IQP.

\paragraph{Instantaneous Quantum Polynomial Circuits}
IQP circuits apply diagonal gates in the computational ($Z$) basis, sandwiched by Hadamard layers~\cite{shepherd2008instantaneous}.
Alternatively, they can be defined as being diagonal in the $X$ basis, omitting the Hadamard layers.
Classical hardness of IQP sampling is supported by complexity-theoretic arguments~\cite{bremner2016average}, and their 
parity structure produces a distributional fingerprint qualitatively distinct from the Clifford family.

\subsection{Classical Shadows}

Classical shadows~\cite{huang2020predicting} provide an efficient protocol for estimating many properties of a quantum state from a single randomized measurement dataset.
For each measurement shot, a random single-qubit unitary $U_i$ is applied to each qubit $i$ before measurement in the computational basis, and the rotation choice is recorded alongside the outcome.
The measurement channel can then be inverted classically to produce unbiased estimators for arbitrary observables.

Variants of this protocol have been proposed to reduce estimator variance by biasing measurements toward a known Hamiltonian structure \cite{hadfield2022measurements}, 
a principle that motivates the variance analysis in Section~\ref{sec:theory}. Efficient classical data structures for shadow measurement under shallow circuits have also been studied \cite{hillmich2021decision}, 
with decision diagrams shown to reduce the classical post-processing cost of shadow estimation.

For a single-qubit Pauli $P_i \in \{X, Y, Z\}$, the shadow estimator is:
\begin{equation}
\hat{P}_i = 3 \cdot (-1)^{b_i} \cdot \mathbf{1}[U_i \text{ matches } P_i\text{'s basis}],
\label{eq:shadow_estimator}
\end{equation}
where the factor of 3 corrects for the $1/3$ probability of any given basis being selected.
For a two-qubit correlator $P_i P_j$, the estimator is the product of single-qubit estimators, yielding a factor of $9$ and requiring both $U_i$ and $U_j$ to match their respective Pauli bases.

The connected correlator $\langle P_i P_j \rangle_c = \langle P_i P_j \rangle - \langle P_i \rangle \langle P_j \rangle$ can be estimated by combining single- and two-qubit shadow estimators.
For $K$ Pauli operators, the required number of shadows scales as $O(K / \varepsilon^2)$ for bounded observables, independent of system size~\cite{huang2020predicting}.
For all $6 \cdot \binom{\nqubits}{2}$ two-qubit Pauli correlators, this yields a total shot cost of $O(\nqubits^2 / \varepsilon^2)$, polynomial in $\nqubits$.

\begin{figure*}[!t]
\centering
\subfloat[Clifford]{\includegraphics[width=0.25\textwidth,height=4cm,keepaspectratio]{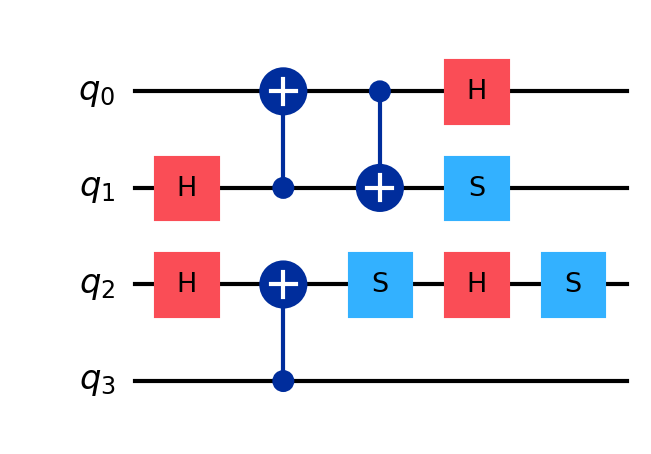}%
\label{fig:clifford}}
\hfil
\subfloat[Clifford$+T$]{\includegraphics[width=0.25\textwidth,height=4cm,keepaspectratio]{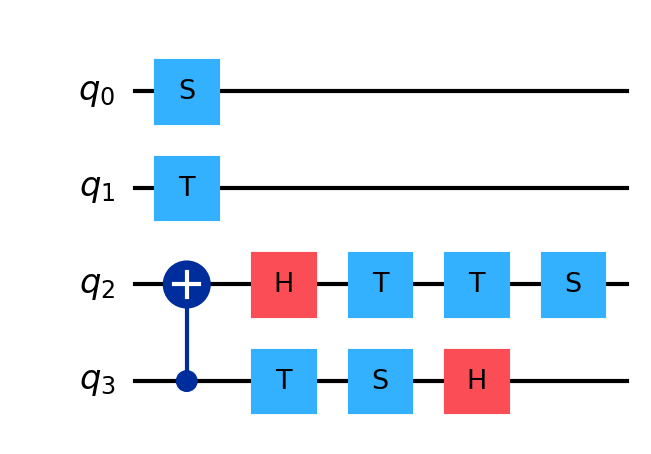}%
\label{fig:cliffordT}}
\hfil
\subfloat[IQP]{\includegraphics[width=0.35\textwidth,height=4cm,keepaspectratio]{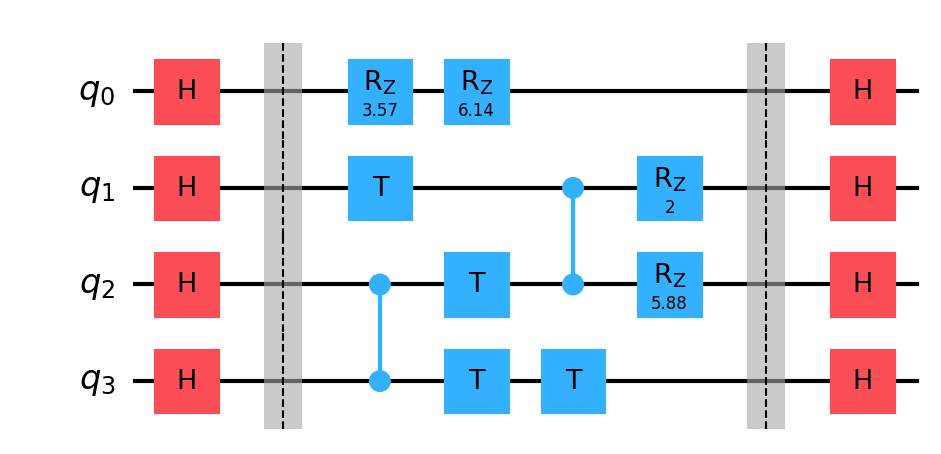}%
\label{fig:iqp}}
\caption{We illustrate example Clifford, Clifford$+T$, and IQP circuits for the reader with four
qubits and $N_c$ = 10. (a) The Clifford circuit is comprised of gates from the Clifford gate set
as defined above. (b) We observe a disproportionate amount of $T$ gates for the Clifford$+T$ example. This is done to 
ensure distinctness from the Clifford circuit family. (c) We show an example IQP circuit. It is 
clear that $N_c$ for IQP circuits refers to the number of operations between the Hadamard gates.}
\label{fig:circuit_families}
\end{figure*}

\subsection{Distribution Classification}

Distribution testing and quantum state discrimination provide the theoretical foundations for circuit family classification.
Classical distribution testing asks whether two sample sequences were drawn from the same distribution, with sample complexity that depends on the alphabet size and closeness notion~\cite{Canonne2020}.
In the quantum setting, quantum state discrimination minimizes the probability of misidentifying one of several candidate states given a fixed number of copies~\cite{Barnett2009}, while quantum hypothesis testing establishes error exponents for distinguishing between two states in the asymptotic copy limit~\cite{Audenaert2007}.
These frameworks are information-theoretically tight but assume access to the full quantum state; when only classical measurement outcomes are available the problem reduces to distinguishing between high-dimensional probability distributions from samples.

More directly related is the body of work on circuit fingerprinting and spoofing detection.
Aaronson and Chen~\cite{AaronsonChen2017} propose linear cross-entropy benchmarking as a statistical test for quantum supremacy claims, exploiting the anticorrelation between a classical simulator's probability assignments and typical bitstring frequencies.
Barak et al.~\cite{BarakEtAl2021} show that low-degree polynomial estimators cannot efficiently distinguish random quantum circuit outputs from uniform noise at large system size, establishing hardness results for feature-based discrimination.
Hangleiter et al.~\cite{HangleiterEtAl2019} study anticoncentration and show that IQP and random circuit families exhibit qualitatively different concentration behavior, a property reflected in Hamming-weight statistics.
The present work complements these results by providing a systematic empirical comparison of four measurement strategies ($Z$-only, nearest-neighbor $ZZ$, multi-basis, and classical shadows) 
for circuit family classification, and is, to our knowledge, the first study to directly benchmark all four strategies against 
each other on the IQP, Clifford, and Clifford$+T$ families across a range of qubit counts.

\section{Methods}
\label{sec:methods}

\subsection{Circuit Generation}

We use the Qiskit SDK \cite{qiskit2024} for all of our research. Using Qiskit, we craft random circuits 
for each of the three previously mentioned circuit families, Clifford, Clifford$+T$, and IQP.
For the first two groups, we define a value $N_c$, which is the number of gates in a given circuit. We
then apply $N_c$ randomly selected gates from the appropriate circuit family to act on randomly selected 
qubits. For the Clifford family of circuits, each gate has an equal chance of being selected. For the 
Clifford$+T$ family, the gates have chances of being selected as outlined in Table~\ref{tab:cliffordT_gate_distros}.
The qubit(s) each gate acts on are selected randomly with no preference to any particular qubit, except 
in the case of two-qubit gates, where the control and target must be different.
The elevated $T$-gate probability of $0.4$ is chosen so that the expected $T$ count per circuit ($0.4\,N_c = 400$) places sampled Clifford$+T$ circuits far outside the regime in which stabilizer-based simulation is efficient, since classical simulation cost grows exponentially with the number of $T$ gates~\cite{bravyi2016improved}; a sensitivity sweep over the $T$-gate density is left to future work.

\begin{table}[htbp]
\centering
\caption{Clifford$+T$ gate selection probabilities. All other families used equal probabilities.}
\label{tab:cliffordT_gate_distros}
\begin{tabular}{|c|c|}\hline
\textbf{Gate} & \textbf{Probability} \\ \hline
 S & 0.2 \\ \hline
 CNOT & 0.2 \\ \hline
 H & 0.2 \\ \hline
 T & 0.4 \\ \hline
\end{tabular}
\end{table}

IQP circuits are created similarly. The only difference, beside the 
gate set used, is that the operator $H^{\otimes \nqubits}$ is applied at the beginning and end of the circuit, as 
shown in Figure~\ref{fig:circuit_families}, where $n$ is the number of qubits. We scale the number of qubits from 4 up to 20 qubits. At each qubit count, we generate 1000 circuits of each 
type. IQP circuits are defined by any combination of $Z$-diagonal gates. In this work, we specifically use $\{T, RZ, \mathrm{CZ}\}$. The rotation angle for the RZ gate is determined randomly 
for each instance of the gate in each circuit.

\begin{table}[htbp]
\centering
\caption{Experimental parameter sweep.}
\label{tab:sweep}
\resizebox{\columnwidth}{!}{%
\begin{tabular}{ll}
\toprule
\textbf{Parameter} & \textbf{Values} \\
\midrule
Circuit families & IQP, Clifford, Clifford$+T$ \\
Qubit counts & 4--20 \\
Shot budget & $16\nqubits^2$ (quadratic) \\
Circuits per family per qubit count & 1000 \\
Operations per circuit & 1000 \\
Simulator & Qiskit Aer (statevector, GPU) \\
Measurement strategies & $Z$-only, NN, multi-basis, shadows \\
Train/test splits & 10 repeated random 80/20 \\
\bottomrule
\end{tabular}}
\end{table}

\subsection{Shot Budget Design}
\label{sec:shots}

The estimation cost of a statistical feature determines the minimum shot budget needed to estimate it reliably.

\paragraph{Polynomial-cost features}
The Hamming-weight histogram bins all $2^\nqubits$ outcomes into $\nqubits + 1$ weight classes, requiring $O(\nqubits / \varepsilon^2)$ shots.
Single-qubit marginals $P(x_i = 1)$ and parity bias require $O(\nqubits / \varepsilon^2)$ and $O(1/\varepsilon^2)$ shots respectively.
Connected two-qubit Pauli correlations across all $\binom{\nqubits}{2}$ pairs require $O(\nqubits^2 / \varepsilon^2)$ shots in total for any single correlator type (e.g., $ZZ$, $XX$, $YY$).
By the classical shadows result, \emph{all} $6\binom{\nqubits}{2}$ two-qubit Pauli correlators ($ZZ$, $XX$, $YY$, $XY$, $XZ$, $YZ$) can be estimated simultaneously at the same $O(\nqubits^2 / \varepsilon^2)$ cost using the shadow protocol.

\paragraph{Excluded exponential-cost features}
The purpose of this work is to classify quantum circuit families using only
polynomial measurement resources.
If exponential resources were available, full quantum state tomography could exactly
reconstruct the output distribution, rendering classification trivial.
In that vein, several prominent statistical measures are excluded from our feature set
due to their exponential shot requirements.
Shannon entropy~\cite{bromiley2004shannon} $H = -\sum_x \hat{p}(x)\log_2 \hat{p}(x)$
requires $O(2^\nqubits / \nqubits)$ shots because the empirical frequency estimator is biased whenever bitstrings go unobserved.
Collision probability $\mathcal{C} = \sum_x p(x)^2$ requires $\Theta(2^{\nqubits/2})$
shots~\cite{Canonne2020}.
R\'{e}nyi entropies of order $\alpha \neq 1$ require at minimum
$\Theta(2^{\nqubits/2})$ shots for integer $\alpha > 1$ and up to
$\Theta(2^\nqubits)$ shots for non-integer $\alpha > 1$ or $\alpha < 1$~\cite{Archarya2017estimating},
all exponential in the number of qubits $\nqubits$.
Total variation distance to uniformity requires observing a representative fraction
of the $2^\nqubits$ outcomes, also at exponential cost.
All such features are incompatible with the polynomial-resource constraint central
to this work and are therefore disregarded.

\paragraph{Quadratic shot scaling}
The most expensive retained features are the two-qubit Pauli correlators, requiring $O(\nqubits^2)$ shots.
We adopt a quadratic shot scaling law $\shots = \lambda \nqubits^2$, where $\shots$ is the total shot budget, $\nqubits$ is the number of qubits, and $\lambda$ is a constant prefactor controlling the statistical quality of the estimators.
We arbitrarily set $\lambda = 16$, giving 256 shots at $\nqubits = 4$ and 6400 shots at $\nqubits = 20$.
For the multi-basis strategy, each of the three basis runs receives $\shots / 3$ shots; the total shot budget remains $\shots = \lambda\nqubits^2$.

\subsection{Measurement Strategies}
\label{sec:strategies}

We compare four measurement protocols, each operating within the same $\shots = 16\nqubits^2$ total shot budget.

\paragraph{$Z$-only baseline}
All shots are taken in the computational basis.
Features extracted: Hamming-weight histogram, $Z$-basis single-qubit marginals, parity bias, and all-pairs connected $ZZ$ correlations.
This is the minimal polynomial-resource strategy and serves as the baseline.

\paragraph{Multi-basis}
Shots are divided equally among three global basis rotations: identity ($Z$ basis), Hadamard on all qubits ($X$ basis), and $S^\dagger H$ on all qubits ($Y$ basis), with $\shots/3$ shots per basis.
Features extracted: Hamming-weight histogram, $Z$/$X$/$Y$ single-qubit marginals, parity bias, and all-pairs connected $ZZ$, $XX$, and $YY$ correlations.
This gives three of the six Pauli pair correlator types at the same total shot cost, at the expense of $3\times$ fewer shots per basis run.
Mixed-basis correlators ($XY$, $XZ$, $YZ$) are not accessible with global rotations.

\paragraph{Classical shadows}
For each of the $\shots$ measurement shots, a random basis ($Z$, $X$, or $Y$) is independently selected for each qubit, and the circuit is measured after applying the corresponding single-qubit rotation.
All six two-qubit Pauli correlator types ($ZZ$, $XX$, $YY$, $XY$, $XZ$, $YZ$) are estimated from the resulting shadow dataset using the unbiased shadow estimators described in Section~\ref{sec:background}.
This achieves the full Pauli correlation structure at the same total shot budget as the $Z$-only baseline.

\paragraph{Nearest-neighbor $ZZ$}
All shots are taken in the computational basis, identical to the $Z$-only strategy.
However, only the $\nqubits - 1$ adjacent qubit pairs are used for connected $ZZ$ correlations rather than all $\binom{\nqubits}{2}$ pairs.
Features extracted: Hamming-weight histogram, $Z$-basis single-qubit marginals, parity bias, and nearest-neighbor connected $ZZ$ correlations.
This reduces the feature dimension from $O(\nqubits^2)$ to $O(\nqubits)$ and tests whether the discriminative signal is concentrated in local, spatially adjacent correlations or requires long-range pairs.

\subsection{Feature Extraction}
\label{sec:features}

\paragraph{$Z$-only features}
The feature vector concatenates Hamming-weight histogram ($\nqubits+1$ values), $Z$-marginals ($\nqubits$ values), parity bias (1 value), and all-pairs $ZZ$ connected correlations ($\binom{\nqubits}{2}$ values):
\begin{equation}
d_{\text{Z}} = (\nqubits+1) + \nqubits + 1 + \tfrac{\nqubits(\nqubits-1)}{2} = \tfrac{\nqubits^2 + 3\nqubits + 4}{2}.
\end{equation}
This grows as $O(\nqubits^2)$, dominated by the all-pairs $ZZ$ term.

\paragraph{Multi-basis features}
The feature vector extends the $Z$-only set by appending $X$-marginals and all-pairs $XX$ correlations, then $Y$-marginals and all-pairs $YY$ correlations:
\begin{equation}
d_{\text{MB}} = d_{\text{Z}} + 2\nqubits + \nqubits(\nqubits-1) = O(\nqubits^2).
\end{equation}

\paragraph{Shadow features}
The feature vector concatenates $Z$/$X$/$Y$ single-qubit expectations ($3\nqubits$ values) and all six Pauli pair connected correlations across all $\binom{\nqubits}{2}$ pairs ($6 \cdot \binom{\nqubits}{2}$ values):
\begin{equation}
d_{\text{S}} = 3\nqubits + 6 \cdot \tfrac{\nqubits(\nqubits-1)}{2} = 3\nqubits^2.
\end{equation}

\paragraph{Nearest-neighbor $ZZ$ features}
The feature vector matches the $Z$-only set but replaces all-pairs $ZZ$ correlations with only the $\nqubits - 1$ adjacent-pair $ZZ$ correlations:
\begin{equation}
d_{\text{NN}} = (\nqubits+1) + \nqubits + 1 + (\nqubits - 1) = 3\nqubits + 1.
\end{equation}
This feature vector scales as $O(\nqubits)$, linear in the number of qubits, in contrast to the $O(\nqubits^2)$ scaling of the other three strategies.

\subsection{Classification Pipeline}

We evaluate four classifiers from scikit-learn: Logistic Regression (L-BFGS, \texttt{max\_iter}=1000), Decision Tree (unlimited depth), Random Forest (100 estimators), and Support Vector Machine (RBF kernel).
For each qubit count and strategy, accuracy is estimated over 10 repeated random 80/20 train/test splits; the mean is reported and the standard deviation is shown as error bars in all figures.
Each classifier is independently trained and evaluated for each measurement strategy and qubit count.

\paragraph{Logistic Regression}
Logistic Regression models the log-odds of class membership as a linear function of the input features.
Given feature vector $\mathbf{x}$, the probability assigned to class $k$ is
$P(y=k \mid \mathbf{x}) = \mathrm{softmax}(\mathbf{W}\mathbf{x} + \mathbf{b})_k$,
and parameters $\mathbf{W}, \mathbf{b}$ are fit by minimizing cross-entropy via L-BFGS.
The decision boundary is a hyperplane in feature space, making the model fast to train and easy to interpret, but limited to linearly separable classes.

\paragraph{Decision Tree}
A Decision Tree recursively partitions the feature space by selecting, at each node, the feature threshold that maximises the reduction in Gini impurity among the training samples reaching that node.
Predictions for a new sample are made by routing it through the learned splits to a leaf, where the majority class label is returned.
With unlimited depth, the tree can perfectly partition any finite training set, capturing highly non-linear boundaries at the cost of a tendency to overfit.
Unlike Logistic Regression, the decision boundary is piecewise axis-aligned rather than a single hyperplane.

\paragraph{Random Forest}
A Random Forest is an ensemble of decision trees, each trained on an independently drawn bootstrap sample of the training data.
At every split, each tree considers only a random subset of features, decorrelating the trees and reducing variance relative to a single deep tree.
The final prediction is the majority vote across all 100 trees.
Random Forest therefore inherits the non-linear expressivity of individual trees while substantially improving generalisation, at the cost of reduced interpretability compared to a single tree.

\paragraph{Support Vector Machine}
A Support Vector Machine (SVM) finds the maximum-margin separating hyperplane in a (possibly kernel-lifted) feature space.
We use the Radial Basis Function (RBF) kernel $K(\mathbf{x}, \mathbf{x}') = \exp(-\gamma\|\mathbf{x}-\mathbf{x}'\|^2)$, which implicitly maps inputs into an infinite-dimensional space, allowing highly non-linear decision boundaries.
Only the support vectors, the training samples closest to the margin, determine the boundary, making SVMs effective in high-dimensional settings where the number of features is large relative to the number of samples.

\paragraph{Key differences}
The four classifiers span two primary axes of variation.
First, \emph{linearity}: Logistic Regression assumes a linear boundary, while the other three can represent non-linear separations.
Second, \emph{variance--bias trade-off}: a single Decision Tree has low bias but high variance; Random Forest reduces variance through ensemble averaging; SVMs control both through the margin objective and the choice of kernel bandwidth.
Evaluating all four allows us to assess whether the circuit family discrimination task is linearly separable in the chosen feature space, or whether the non-linear structure of the classifiers is necessary at larger qubit counts.

\section{Results}
\label{sec:results}

\subsection{Measurement Strategy Comparison}

Figures~\ref{fig:acc_z_only}--\ref{fig:acc_nn} show classifier accuracy as a function 
of qubit count for each measurement strategy. Error bars show the standard deviation 
over 10 repeated random 80/20 train/test splits. Table~\ref{tab:peak_accuracies} 
summarizes the peak accuracy achieved by each classifier under each measurement strategy.

\begin{table*}[t]
\centering
\caption{Peak classifier accuracy by measurement strategy.}
\label{tab:peak_accuracies}
\begin{tabular}{lcccc}
\toprule
Strategy & Random Forest & Decision Tree & SVM & Logistic Regression \\
\midrule
$Z$-only          & \textbf{0.91} & \textbf{0.80} & 0.52          & 0.39 \\
NN ($ZZ$)         & 0.89          & 0.78          & \textbf{0.56} & 0.36 \\
Multi-basis       & 0.85          & 0.70          & 0.49          & 0.36 \\
Classical shadows & 0.67          & 0.50          & 0.52          & 0.39 \\
\bottomrule
\end{tabular}
\end{table*}

The ordering $Z$-only $\approx$ NN $>$ multi-basis $>$ shadows is consistent across all classifiers and all qubit counts where any strategy achieves above-chance accuracy.
The near-equivalence of $Z$-only and NN is a key finding: reducing the $ZZ$ feature set from all $\binom{\nqubits}{2}$ pairs to only the $\nqubits - 1$ adjacent pairs costs less than $0.02$ in Random Forest accuracy, indicating that the discriminative signal is concentrated in local, nearest-neighbor correlations.
This result inverts the original hypothesis that richer Pauli correlation structure would improve classification.

\begin{figure}[htbp]
    \centering
    \includegraphics[width=\linewidth]{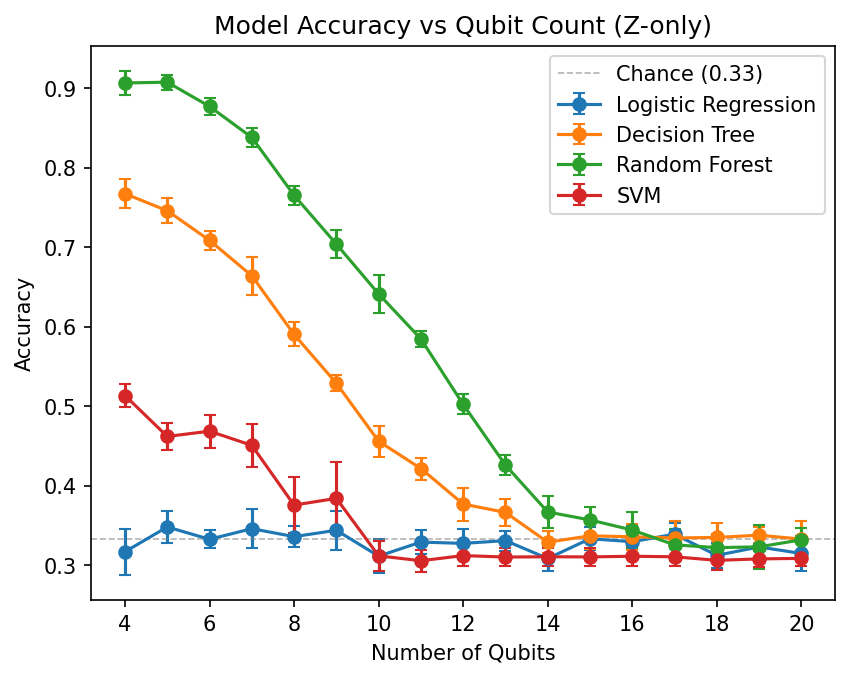}
    \caption{Classifier accuracy vs.\ qubit count under the $Z$-only measurement strategy (4--20 qubits). Error bars show std over 10 repeated train/test splits.}
    \label{fig:acc_z_only}
\end{figure}

\begin{figure}[htbp]
    \centering
    \includegraphics[width=\linewidth]{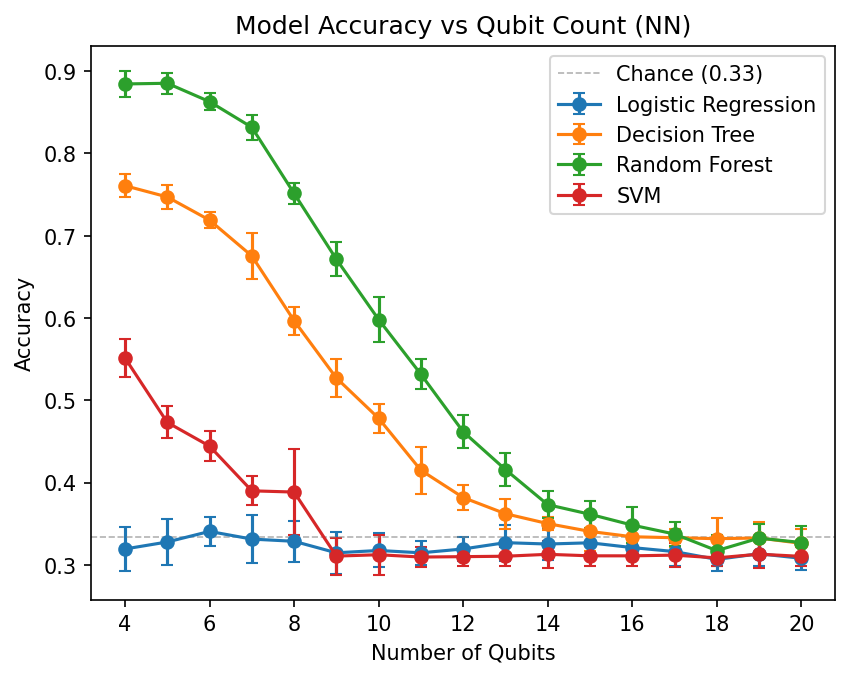}
    \caption{Classifier accuracy vs.\ qubit count under the nearest-neighbor $ZZ$ measurement strategy (4--20 qubits). Error bars show std over 10 repeated train/test splits.}
    \label{fig:acc_nn}
\end{figure}

\begin{figure}[htbp]
    \centering
    \includegraphics[width=\linewidth]{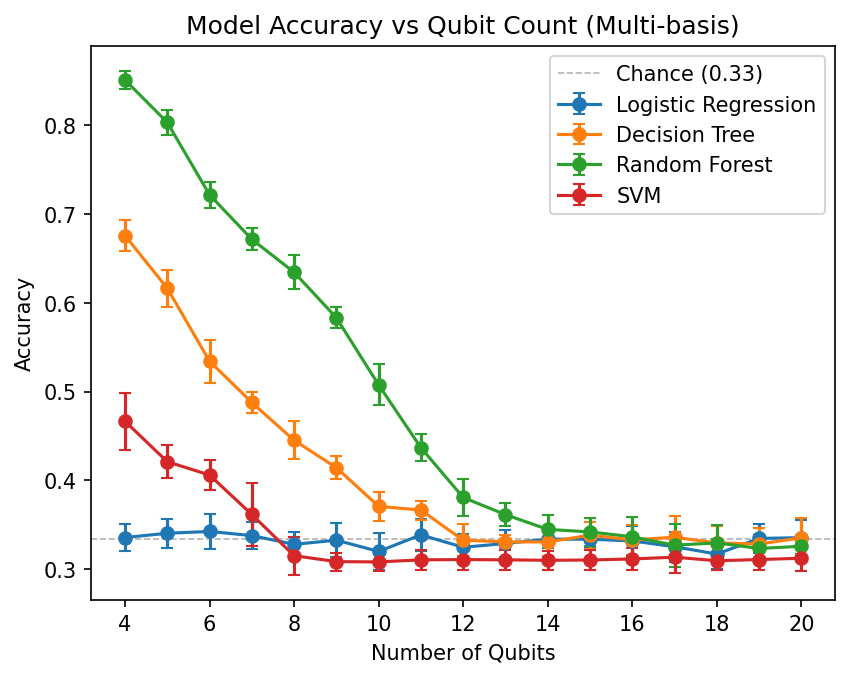}
    \caption{Classifier accuracy vs.\ qubit count under the multi-basis measurement strategy (4--20 qubits). Error bars show std over 10 repeated train/test splits.}
    \label{fig:acc_multi_basis}
\end{figure}

\begin{figure}[htbp]
    \centering
    \includegraphics[width=\linewidth]{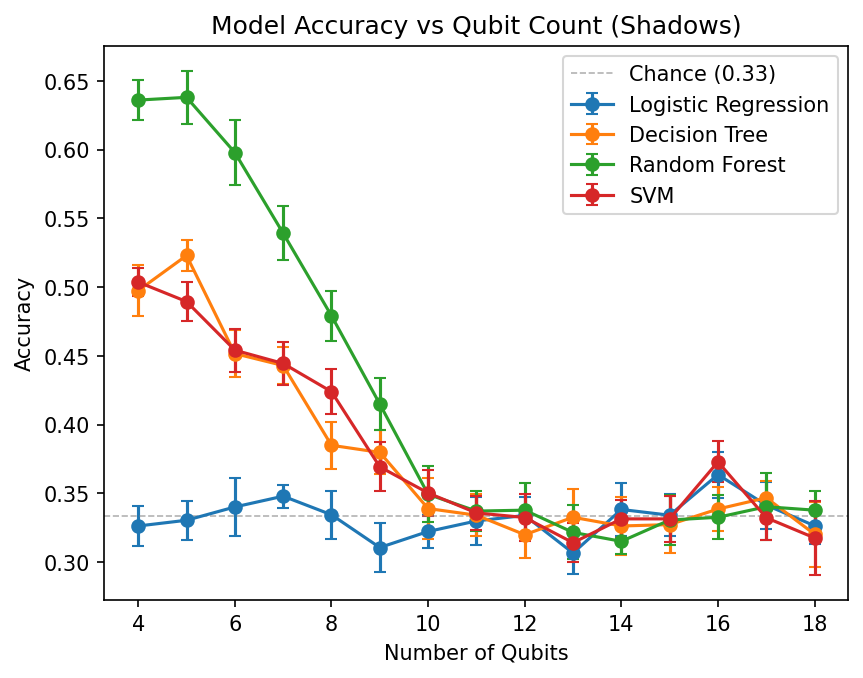}
    \caption{Classifier accuracy vs.\ qubit count under the classical shadows measurement strategy (4--18 qubits). Shadows data is unavailable for 19--20 qubits due to incomplete simulation runs. Error bars show std over 10 repeated train/test splits.}
    \label{fig:acc_shadows}
\end{figure}

\subsection{Scaling Collapse}

All four strategies show a clear degradation in accuracy as qubit count increases, with a collapse to near-chance performance ($\approx 0.33$) above approximately 12--14 qubits.
This collapse is visible in all four classifiers under all four measurement strategies: even the best-performing Random Forest under $Z$-only reaches chance accuracy by 14 qubits.
Below 10 qubits, tree-based methods (Decision Tree and Random Forest) maintain clearly above-chance accuracy under $Z$-only and NN; above 12 qubits no classifier reliably exceeds chance under any strategy.
The collapse is consistent with the quadratic shot budget $\shots = 16\nqubits^2$ becoming insufficient to resolve the class-separating correlations as system size grows: the same constant prefactor $\lambda = 16$ that provides adequate statistics at $\nqubits = 4$ (256 shots) is no longer sufficient at $\nqubits = 14$ (3136 shots) where the number of all-pairs correlators has grown as $\binom{\nqubits}{2}$.
Notably, the NN strategy uses only $O(\nqubits)$ features, so each nearest-neighbor correlator receives more shots per estimate under the same total budget; this does not prevent the scaling collapse, suggesting the bottleneck is distributional overlap between families rather than estimator variance alone.

\subsection{Linear Separability}

Logistic Regression is flat at approximately $0.36$--$0.39$ across all four strategies and all qubit counts, performing at or near the three-class chance level throughout.
This indicates that the class boundary separating IQP, Clifford, and Clifford$+T$ circuits is nonlinear in the Pauli correlator feature space: a single hyperplane cannot separate the three families regardless of how many correlator types are included in the feature vector.
The strong performance of Random Forest and Decision Tree relative to Logistic Regression and SVM confirms that the discriminative structure is captured by axis-aligned partitions in feature space rather than by a linear function of the features.
SVM with the RBF kernel performs between the linear and tree-based methods, indicating that the kernel-induced feature mapping captures some but not all of the relevant nonlinear structure.

\section{Theoretical Framework}
\label{sec:theory}

The empirical ordering $Z$-only $\approx$ NN $>$ multi-basis $>$ shadows could in principle be an artifact of the specific shot budget, gate set, or circuit depth used in our experiments.
In this section we show that the core of this ordering is structural, by proving two exact results.
First, a frame identity for IQP circuits (Lemma~\ref{lem:frame}) implies that \emph{every} purely-$X$ observable of an IQP output state vanishes identically (Proposition~\ref{prop:xnull}): the $X$-sector features purchased by the multi-basis and shadow protocols are deterministic zeros on one of the three classes.
Second, at equal total shot budget, the three global protocols estimate any given $ZZ$ correlator with variances in the exact ratio $1 : 3 : {\geq}9$ ($Z$-only : multi-basis : shadows), as shown in Proposition~\ref{prop:variance_ratio} and Corollary~\ref{cor:variance_ordering}.
Neither result requires assumptions on circuit depth, gate probabilities, or qubit count.
Combined with the empirical observation that the remaining purchased sectors do not compensate (Section~\ref{sec:results}), these results explain the observed ordering; we are explicit at the end of the section about which parts are theorems and which remain empirical.

\subsection{A Frame Identity and the X-Sector Nullity of IQP Circuits}

Every IQP circuit has the form $C = H^{\otimes\nqubits} D\, H^{\otimes\nqubits}$, where $D$ is a product of gates diagonal in the computational basis (here drawn from $\{T, RZ, \mathrm{CZ}\}$).
Because conjugation by the boundary Hadamard layers exchanges $X$- and $Z$-type Paulis, measurements on the output state correspond exactly to observables of the interior state $D|{+}^{\nqubits}\rangle$, where $|{+}^{\nqubits}\rangle = H^{\otimes\nqubits}|0^{\nqubits}\rangle$ and all of the diagonal phase structure resides.

\begin{lemma}[Hadamard Frame Identity]
\label{lem:frame}
Let $C = H^{\otimes\nqubits} D\, H^{\otimes\nqubits}$ with $D$ diagonal in the computational basis, and let $|\psi\rangle = C|0^{\nqubits}\rangle$.
For any nonempty $S \subseteq \{1, \dots, \nqubits\}$, write $Z_S = \prod_{i \in S} Z_i$ and $X_S = \prod_{i \in S} X_i$. Then
\begin{align}
\langle\psi|\, Z_S\, |\psi\rangle &= \langle{+}^{\nqubits}|\, D^\dagger X_S D\, |{+}^{\nqubits}\rangle, \label{eq:frame_z}\\
\langle\psi|\, X_S\, |\psi\rangle &= \langle{+}^{\nqubits}|\, D^\dagger Z_S D\, |{+}^{\nqubits}\rangle. \label{eq:frame_x}
\end{align}
\end{lemma}

\begin{proof}
Using $H^{\otimes\nqubits}|0^{\nqubits}\rangle = |{+}^{\nqubits}\rangle$ and $H^{\otimes\nqubits} Z_S\, H^{\otimes\nqubits} = X_S$,
\begin{align*}
\langle\psi|\, Z_S\, |\psi\rangle
&= \langle 0^{\nqubits}|\, H^{\otimes\nqubits} D^\dagger \bigl(H^{\otimes\nqubits} Z_S\, H^{\otimes\nqubits}\bigr) D\, H^{\otimes\nqubits}\, |0^{\nqubits}\rangle \\
&= \langle{+}^{\nqubits}|\, D^\dagger X_S D\, |{+}^{\nqubits}\rangle.
\end{align*}
Eq.~\eqref{eq:frame_x} follows symmetrically from $H^{\otimes\nqubits} X_S\, H^{\otimes\nqubits} = Z_S$.
\end{proof}

Eq.~\eqref{eq:frame_z} states that the computational-basis statistics of the IQP output are exactly the $X$-sector statistics of the interior diagonal state; indeed, the full $Z$-basis output distribution of $C$ equals the $X$-basis distribution of $D|{+}^{\nqubits}\rangle$.
$Z$-only measurement therefore probes the diagonal phase structure of the IQP family directly.
Eq.~\eqref{eq:frame_x} has a sharper consequence.

\begin{proposition}[$X$-Sector Nullity of IQP Circuits]
\label{prop:xnull}
For every IQP circuit as above and every nonempty $S$,
\[
\langle\psi|\, X_S\, |\psi\rangle = 0.
\]
In particular, all $X$ marginals $\langle X_i\rangle$, all correlators $\langle X_i X_j\rangle$, and all connected correlators $\langle X_i X_j\rangle_c$ of the IQP output state are exactly zero, for every realization of the diagonal layer $D$.
\end{proposition}

\begin{proof}
$D$ is diagonal in the computational basis and therefore commutes with $Z_S$. By Eq.~\eqref{eq:frame_x},
\[
\langle\psi|X_S|\psi\rangle = \langle{+}^{\nqubits}|D^\dagger Z_S D|{+}^{\nqubits}\rangle = \langle{+}^{\nqubits}|Z_S|{+}^{\nqubits}\rangle = 0,
\]
since $\langle{+}|Z|{+}\rangle = 0$ on every site of $S$. The statement for connected correlators follows because each term vanishes.
\end{proof}

\begin{remark}[Scope of the Nullity Result]
\label{rem:nullity}
Proposition~\ref{prop:xnull} constrains only the pure-$X$ sector; $Y$-type and mixed-basis correlators of IQP outputs are not identically zero in general.
Nor does it say that $X$-sector features are useless in principle: a feature that is deterministically $0$ on IQP but nonzero on Clifford or Clifford$+T$ outputs could in principle help separate the classes.
What it establishes is that, on the IQP class, every shot spent estimating an $X$-sector feature measures a known constant and therefore contributes sampling noise but no information about the circuit instance.
Whether the nonzero $X$- and $Y$-sector values of the other two families add discriminative power beyond the $Z$ sector is an empirical question, and Table~\ref{tab:peak_accuracies} answers it in the negative for this task: the two protocols that pay for broader sector access with a reduced effective $Z$-sector budget (multi-basis and shadows) underperform $Z$-only across all classifiers.
Proposition~\ref{prop:xnull} also yields a falsifiable consistency check on the measurement pipeline: estimated $X$ marginals and $XX$ correlators of simulated IQP circuits must be statistically indistinguishable from zero at every qubit count.
\end{remark}

\subsection{Estimator Variance at Equal Shot Budget}

We now quantify the statistical price at which each protocol estimates the $ZZ$ correlators that carry the discriminating signal.
(The NN strategy uses the identical per-shot estimator as $Z$-only on a subset of pairs and needs no separate analysis.)
The $3^k$-type sampling overhead of uniformly random Pauli shadows for a fixed $k$-local observable is well known and motivates derandomized and locally biased shadow schemes~\cite{huang2020predicting,hadfield2022measurements}; Proposition~\ref{prop:variance_ratio} records the exact constants for the three protocols compared in this work.

\begin{proposition}[Per-Correlator Variance at Equal Budget]
\label{prop:variance_ratio}
Fix a qubit pair $(i, j)$, let $x = \langle Z_i Z_j\rangle^2$, and give each protocol the same total budget of $s$ shots.
Each protocol admits an unbiased estimator of $\langle Z_i Z_j\rangle$ with variance
\[
\mathrm{Var}^{(Z)} = \frac{1 - x}{s}, \quad
\mathrm{Var}^{(\mathrm{MB})} = \frac{3(1 - x)}{s}, \quad
\mathrm{Var}^{(S)} = \frac{9 - x}{s},
\]
so that
\[
\mathrm{Var}^{(Z)} : \mathrm{Var}^{(\mathrm{MB})} : \mathrm{Var}^{(S)} \;=\; 1 : 3 : \frac{9 - x}{1 - x},
\qquad \frac{9 - x}{1 - x} \geq 9,
\]
with equality on the right if and only if $x = 0$.
\end{proposition}

\begin{proof}
Under $Z$-only, each of the $s$ computational-basis shots yields the estimator $\hat{o}^{(Z)}_t = (-1)^{b_i + b_j}$ with $(\hat{o}^{(Z)}_t)^2 = 1$, so the per-shot variance is $1 - x$ and the $s$-shot mean has variance $(1 - x)/s$.
Under multi-basis, the identical estimator is available only on the $s/3$ shots taken in the $Z$ basis, giving variance $3(1 - x)/s$.
Under the random-Pauli shadow protocol of Eq.~\eqref{eq:shadow_estimator}, the per-shot estimator is $\hat{o}^{(S)}_t = 9\,(-1)^{b_i + b_j}\, \mathbf{1}[U_i = Z]\, \mathbf{1}[U_j = Z]$, nonzero with probability $\tfrac{1}{9}$; it is unbiased with $\mathbb{E}[(\hat{o}^{(S)}_t)^2] = 81 \cdot \tfrac{1}{9} = 9$, so the per-shot variance is $9 - x$ and the $s$-shot variance is $(9 - x)/s$.
The ratio bound follows from $9 - x \geq 9(1 - x) \Leftrightarrow 8x \geq 0$.
\end{proof}

Analogous factors hold for the single-qubit terms $\langle Z_i \rangle$ entering the connected correlators, with variance ratios $1 : 3 : {\geq}3$.

\begin{corollary}[Unconditional Variance Ordering]
\label{cor:variance_ordering}
For every output state and every qubit pair, at equal total shot budget,
\[
\mathrm{Var}^{(Z)} \;\leq\; \mathrm{Var}^{(\mathrm{MB})} \;\leq\; \mathrm{Var}^{(S)},
\]
with $\mathrm{Var}^{(\mathrm{MB})}/\mathrm{Var}^{(Z)} = 3$ exactly, $\mathrm{Var}^{(S)}/\mathrm{Var}^{(Z)} \geq 9$, and $\mathrm{Var}^{(S)}/\mathrm{Var}^{(\mathrm{MB})} \geq 3$.
The ordering requires no assumption on the circuit family, and a large correlator only widens the shadow penalty: $(9 - x)/(1 - x)$ is increasing in $x$ and diverges as $x \to 1$.
\end{corollary}

Together, Proposition~\ref{prop:xnull} and Corollary~\ref{cor:variance_ordering} account for the observed ordering as follows.
Relative to $Z$-only, the multi-basis and shadow protocols accept a $3\times$ and ${\geq}9\times$ variance penalty on every $ZZ$ feature in exchange for access to additional Pauli sectors.
On the IQP class the entire purchased $X$ sector is identically zero (Proposition~\ref{prop:xnull}), and the remaining purchased sectors are empirically uninformative for this task (Table~\ref{tab:peak_accuracies} and Figures~\ref{fig:acc_multi_basis}--\ref{fig:acc_shadows}).
We emphasize the division of labor: the variance ordering and the $X$-sector nullity are theorems; the claim that no purchased sector compensates for the variance penalty is exact for the $X$ sector on the IQP class and empirical for the remainder.
These results explain the ranking of the four strategies evaluated here; they do not assert that $Z$-only measurement is optimal among all polynomial-resource protocols.

\section{Discussion}
\label{sec:discussion}

\subsection{Why Z-Only and NN Outperform Multi-Basis and Shadows}

The central finding --- that $Z$-only outperforms both multi-basis and classical shadows --- is explained by the structure of IQP circuits and formalized by the theoretical framework in Section~\ref{sec:theory}.
IQP circuits are defined by diagonal gates in the computational basis sandwiched between Hadamard layers; by Lemma~\ref{lem:frame}, the computational-basis statistics of an IQP output state coincide exactly with the $X$-sector statistics of the interior diagonal state, so $Z$-basis measurement probes the diagonal phase structure directly.
The all-pairs connected $ZZ$ correlations captured by the $Z$-only strategy therefore probe precisely the interaction structure that defines the IQP family.

Classical shadows, by contrast, dilutes the $Z$-basis shot budget by a factor of approximately $1/3$ per qubit per shot: in the random single-qubit Clifford protocol, each qubit is measured in the $Z$ basis only one-third of the time on average, so each $ZZ$ correlator estimate is derived from roughly $\shots/9$ shots rather than $\shots$ shots.
By Proposition~\ref{prop:variance_ratio}, this is not merely a practical inconvenience but a guaranteed variance penalty: the shadow estimator's per-shot variance for any $ZZ$ correlator is at least $9\times$ that of the $Z$-only estimator, with the penalty largest precisely when the correlator is large.
The shadow protocol compensates by simultaneously estimating five additional correlator types ($XX$, $YY$, $XY$, $XZ$, $YZ$), but the $XX$ correlators among them are identically zero on the IQP class (Proposition~\ref{prop:xnull}), and the remaining types are empirically uninformative for the IQP/Clifford/Clifford$+T$ classification task (Table~\ref{tab:peak_accuracies}).
The result is that shadows pays a statistical cost relative to $Z$-only without receiving a compensating accuracy benefit.
Multi-basis measurements sit between the two: $ZZ$ correlators receive $\shots/3$ shots rather than $\shots$, and the additional $XX$ and $YY$ correlators are similarly uninformative, yielding performance intermediate between $Z$-only and shadows.

\subsection{Nearest-Neighbor \texorpdfstring{$ZZ$}{ZZ} and the Locality of the Discriminative Signal}

The near-equivalence of $Z$-only and NN is the most surprising finding of this comparison.
$Z$-only computes all $\binom{\nqubits}{2}$ pairwise $ZZ$ correlations ($O(\nqubits^2)$ features), while NN retains only the $\nqubits - 1$ adjacent-pair correlations ($O(\nqubits)$ features), yet the Random Forest accuracy gap is less than $0.02$ across the full qubit range.
This indicates that the discriminative signal distinguishing IQP, Clifford, and Clifford$+T$ families is predominantly carried by local, nearest-neighbor $ZZ$ correlations; the long-range pairs included in the $Z$-only feature set contribute negligible additional classification power under this shot budget.

The locality of the signal also has a practical implication: under the same shot budget $\shots = 16\nqubits^2$, each NN correlator estimate is derived from far more shots than each all-pairs correlator estimate in the $Z$-only strategy, since the budget is spread over only $\nqubits - 1$ pairs rather than $\binom{\nqubits}{2}$.
The fact that $Z$-only nonetheless retains a slight edge suggests that the long-range correlations, though individually weak, contribute a small net positive signal rather than acting as pure noise.

\subsection{Why Logistic Regression Fails}

Logistic Regression's flat performance at near-chance accuracy across all strategies indicates that the three circuit families are not linearly separable in the Pauli correlator feature space.
Tree-based methods (Random Forest, Decision Tree) substantially outperform both Logistic Regression and SVM because axis-aligned recursive partitioning can capture the non-linear class boundaries that separate IQP, Clifford, and Clifford$+T$ distributions in this feature space.
The failure of SVM (RBF kernel) to match Random Forest suggests that the discriminative structure is not well-approximated by radial similarity in the correlator space, but is instead captured by threshold conditions on individual features or small feature subsets --- precisely the type of structure that decision trees exploit.

\subsection{Scaling Collapse and Shot Budget}

The collapse of all strategies to chance accuracy above approximately 12--14 qubits reflects a fundamental tension between the quadratic shot budget and the growing dimensionality of the feature space.
As $\nqubits$ increases, the number of all-pairs $ZZ$ correlators grows as $\binom{\nqubits}{2} \approx \nqubits^2/2$, while the shot budget grows as $16\nqubits^2$.
The ratio of shots per correlator therefore remains roughly constant at $\lambda = 16$, but the absolute number of shots per correlator ($\approx 32$ at all qubit counts) is insufficient to resolve the increasingly subtle distributional differences between circuit families as system size grows.
Increasing the prefactor $\lambda$ beyond 16 would shift the collapse threshold to larger qubit counts, at the cost of a proportionally larger total shot budget.

\subsection{Limitations}

This study uses noiseless classical simulation.
Real device noise will alter the output distributions of all three circuit families and likely degrade classification accuracy, potentially in family-dependent ways.
Extending to noisy simulation via Aer noise models, or to real hardware, is a natural next step.

The shot budget prefactor $\lambda = 16$ is heuristic.
A principled choice requires knowing the signal-to-noise ratio of the class-separating correlators, which depends on circuit family, qubit count, and gate depth, and is not known a priori.

The scaling collapse above 12 qubits establishes that a quadratic shot budget $\shots = 16\nqubits^2$ is insufficient for reliable classification at large qubit counts, but this result does not rule out the existence of a more efficient polynomial-resource strategy.
A subquadratic shot budget, a different feature representation, or an adaptive measurement protocol could in principle recover discriminative power at larger system sizes; this work makes no claim about the fundamental hardness of the classification task, only about the specific strategies evaluated here.

The theoretical results of Section~\ref{sec:theory} explain the relative ordering of the four strategies evaluated here; they do not establish optimality of $Z$-only measurement among all polynomial-resource protocols, and the claim that the non-$Z$ Pauli sectors carry no compensating signal is proven only for the $X$ sector on the IQP class and is otherwise empirical.

The Clifford$+T$ gate-selection probabilities (Table~\ref{tab:cliffordT_gate_distros}) fix the $T$-gate density at $0.4$; classification accuracy as a function of this density was not swept and is left to future work.

Classical shadows results are unavailable for 19--20 qubits due to incomplete HPC simulation runs.
The missing data points fall in the regime where all strategies already operate at near-chance accuracy; the scaling collapse characterization is therefore not materially affected by this gap.

\section{Future Work}
\label{sec:future}

Several directions follow naturally from this work.

\paragraph{Theoretical Backing for Identifying IQP}
While some obvious extensions of our work are described below, the most pressing question inspired by our work is the following. 
``Does there exist a BPP algorithm that, given sample access to an unknown distribution, decides with bounded error whether that distribution could have been generated 
by a polynomial-size IQP circuit, as opposed to another class of circuits?'' The collapse to chance accuracy around the 12 qubit mark hints 
that the answer may be no, but our approach is merely one algorithm among many possible polynomial-resource strategies, and the question of whether any 
such strategy exists remains open. Also, the work by Aaronson and Arkhipov~\cite{aaronson2011computational} suggests that the answer may be no. We note, however, 
that Aaronson and Arkhipov consider the harder problem of verifying a device is sampling from the full BosonSampling distribution against arbitrary alternatives. 
Our question, ``distinguishing IQP from specific classically simulable families,'' is more structured, and it remains open whether this additional structure 
makes BPP verification possible or whether the underlying sampling hardness of IQP still precludes it.

\paragraph{Extended circuit families}
IQP circuits are defined by diagonal gates in the $Z$ basis; one can analogously define circuit families diagonal in the $X$ or $Y$ basis by conjugating the diagonal layer with $H^{\otimes \nqubits}$ or $(SH)^{\otimes \nqubits}$, respectively.
These families would carry qualitatively different distributional signatures under $Z$-basis measurement and could serve as additional classification targets, extending the four-strategy comparison evaluated here.
The two-line argument of Lemma~\ref{lem:frame} carries over in the rotated frame, yielding the corresponding sector-nullity statement for each such family and making these extensions a natural testbed for the framework of Section~\ref{sec:theory}.

\paragraph{Noisy simulation and hardware}
All experiments here use noiseless classical simulation.
Realistic device noise will alter the output distributions of all three circuit families in ways that may depend on family identity, qubit count, and gate depth.
Extending to Aer noise models or real hardware would establish whether the $Z$-only advantage and the scaling collapse threshold persist under realistic conditions.

\paragraph{Shot budget scaling}
The collapse to chance accuracy above approximately 12 qubits is tied to the fixed prefactor $\lambda = 16$.
Increasing $\lambda$ would shift the collapse threshold to larger qubit counts; a principled choice of $\lambda$ requires understanding the signal-to-noise ratio of the class-separating correlators as a function of system size, which remains an open problem.
Another approach is to increase the exponent of $\nqubits$ such that the shot scaling is no longer quadratic, but perhaps cubic or quartic. As all of the 
features grow quadratically, this could ensure an appropriate (polynomial) number of shots are used to accurately observe features.

\section{Conclusion}
\label{sec:conclusion}

We compared four polynomial-resource measurement strategies: $Z$-basis-only, nearest-neighbor $ZZ$, multi-basis ($Z$, $X$, $Y$), and classical shadows for classifying IQP, Clifford, and Clifford$+T$ circuit families across 4--20 qubits under a quadratic shot budget $\shots = 16\nqubits^2$.

Contrary to our initial hypothesis, $Z$-only measurements outperform multi-basis and classical shadows across all qubit counts and all four classifiers evaluated, and the $O(\nqubits)$-feature nearest-neighbor $ZZ$ (NN) strategy matches $Z$-only to within $0.02$ in Random Forest accuracy.
Peak Random Forest accuracies are $0.91$ ($Z$-only), $0.89$ (NN), $0.85$ (multi-basis), and $0.67$ (classical shadows).
The near-equivalence of $Z$-only and NN is the most striking finding: reducing the feature set from all $\binom{\nqubits}{2}$ pairwise correlators to only the $\nqubits - 1$ adjacent-pair correlators costs less than $0.02$ in accuracy, indicating that the discriminative signal is spatially local.

We provide a formal theoretical explanation in Section~\ref{sec:theory}. Lemma~\ref{lem:frame} and Proposition~\ref{prop:xnull} show that every purely-$X$ observable of an IQP output state vanishes identically, so the $X$-sector features accessed by multi-basis and shadow measurements are deterministic zeros on the IQP class. Proposition~\ref{prop:variance_ratio} and Corollary~\ref{cor:variance_ordering} establish that, at equal total shot budget, the multi-basis and shadow protocols estimate each $ZZ$ correlator with exactly $3\times$ and at least $9\times$ the variance of $Z$-only measurement, a penalty that only grows with the correlator magnitude. Together these exact results make the observed ordering a consequence of the algebraic structure of IQP circuits and of the estimators themselves, rather than an artifact of the experimental setup, while asserting no optimality claim beyond the four strategies evaluated.

All strategies collapse to near-chance accuracy ($\approx 0.33$) above approximately 12--14 qubits under the quadratic shot budget, and Logistic Regression performs at chance across all strategies, confirming that the class boundaries are nonlinear in the Pauli correlator feature space. Whether any polynomial measurement protocol can classify these families reliably at large qubit counts remains an open question.

\section*{Acknowledgments}
The authors thank the Virginia Modeling, Analysis and Simulation Center (VMASC), the 
Center for Secure and Intelligent Critical Systems (CSICS), and Old Dominion University for institutional support. 
This research was supported by the Research Computing clusters at Old Dominion University.
The Wahab cluster is supported in part by National Science Foundation grant CNS-1828593, ``MRI Acquisition: A Reconfigurable Computing Infrastructure Enabling Interdisciplinary and Collaborative Research in Hampton Roads.''
Claude was used in preparing this manuscript to check for formatting and to proofread for grammar and clarity, ensuring overall readability and coherence of the presentation.
Claude Code was used to assist in writing the theoretical framework section, providing a structured outline and initial drafts of the proofs, which were then refined and verified by the authors.
Claude Code was also used to generate the LaTeX code for the figures and tables, ensuring consistency in formatting and style across the manuscript.
The authors declare no competing interests.
\bibliographystyle{IEEEtran}
\bibliography{references}

\end{document}